  \documentclass[preprint,11pt,3p,authoryear]{elsarticle}

\usepackage{subfigure}
\usepackage{amsthm}
\usepackage{amsmath,amssymb}
\usepackage{graphicx,cite,url}
\usepackage[usenames]{color}
\usepackage{booktabs}
\usepackage{colortbl,dcolumn}
\usepackage{tabularx}
\usepackage{epstopdf}

\biboptions{comma,sort&compress}

\newtheorem{TT}{Theorem}

\newtheorem{DD}{Definition}

\newtheorem{REM}{Remark}

\DeclareMathOperator{\Tr}{Tr}

\usepackage{kbordermatrix}

\usepackage{blkarray}

\usepackage{varwidth}
\usepackage{algorithm}
\usepackage{algpseudocode}
\usepackage{xcolor}

\newcommand{\vect}[1]{\boldsymbol{#1}}

\journal{}

\begin{document}

\begin{frontmatter}

\newcommand{\TITLE}{Langevin method for a continuous stochastic car-following model and its stability conditions}
\newcommand{\AUTHORS}{}%
\newcommand{\AFFILIATIONS}{.}
\newcommand{\SUPPORTS}{}

\author[addr1]{D. Ngoduy\corref{cor1}}
\author[addr1]{S. Lee}
\author[addr2]{M. Treiber}
\author[addr1]{M. Keyvan-Ekbatani}
\author[addr3]{H.L. Vu}
\address[addr1]{Connected Traffic Systems Lab, Dept. Civil and Natural Resources Eng., University of Canterbury, NZ}
\cortext[cor1]{Corresponding author. email:dong.ngoduy@canterbury.ac.nz}
\address[addr2]{Institute for Transport \& Economics, Technical University of Dresden, Germany}
\address[addr3]{Institute of Transport Studies, Monash University, Australia}



\title{\TITLE}


%
%
%
%
%
%

\begin{abstract}
In car-following models, the driver reacts according to his physical and psychological abilities which may change over time.  However, most car-following models are deterministic and do not capture the stochastic nature of human perception. It is expected that purely deterministic traffic models may produce unrealistic results due to the stochastic driving behaviors of drivers. This paper is devoted to the development of a distinct car-following model where a stochastic process is adopted to describe the time-varying random acceleration which essentially reflects the random individual perception of driver behavior with respect to the leading vehicle over time. In particular, we apply coupled Langevin equations to model complex human driver behavior. In the proposed model, an extended Cox-Ingersoll-Ross (CIR) stochastic process will be used to describe the stochastic speed of the follower in response to the stimulus of the leader. An important property of the extended CIR process is to enhance the non-negative properties of the stochastic traffic variables (e.g. non-negative speed) for any arbitrary model parameters. Based on stochastic process theories, we derive stochastic linear stability conditions which, for the first time, theoretically capture the effect of the random parameter on traffic instabilities. Our stability results conform to the empirical results that the traffic instability is related to the stochastic nature of traffic flow at the low speed conditions, even when traffic is deemed to be stable from deterministic models.
		

\end{abstract}


\end{frontmatter}

\section{Introduction}
\label{Introduction}

To date, there have been impressive advances in modeling the traffic-flow dynamics at both microscopic and macroscopic level. However, a majority of existing models are deterministic and fail to describe the existing uncertainty in human perception and driving behavior. These uncertainties are, among many factors, affecting the formation and propagation of stop-and-go waves \citep{Yeo2009}. Recent studies have shown that the stochasticity is a result of many features observed in real-life traffic. For example, the concave and stochastic growth patterns of oscillations are not captured by most deterministic models \citep{Tian2016b}. Emerging technologies in data collection allow us to study the true dynamics of car-following patterns which are not observed by in-situ observations of traffic flow on road. Extensive experiment results in \citet{Jiang2018} have indicated that the interplay between stochastic factors and speed adaptation is vital in the formation and evolution of oscillations. The authors have argued that the traffic instability might be determined by the competition between stochastic factors and the speed adaptation effect. This leads some implications to traffic flow theory: as deterministic models are not able to reproduce such instabilities, many efforts have been taken to develop stochastic models to capture the random human behavior over time. 

At the macroscopic level, which describes the dynamics of traffic flow at an aggregated level of the detail, the stochastic properties of traffic flow have been captured well via the modified fundamental diagrams (i.e. stochastic flow-density relationships). For example, \citet{Sumalee2011} have added random noises in the model parameters of a Stochastic Cell Transmission Model (SCTM) while \citet{Zhong2013} have extended the SCTM for the traffic networks with uncertainties in supply and demand. \citet{Ngoduy2011} has adopted a stochastic fundamental diagram in a multi-class LWR model by adding a random noise in the capacity using Gama distribution. \citet{Li2012} have also considered a stochastic fundamental diagram by adding local noises to the model parameters. \citet{Torduex2014} used a distinct stochastic jump process for the LWR model. Basically, these methods deal with the uncertainties numerically by adding noise with known probability distributions directly to the discrete model. Recently, \citet{Laval2013} have proposed a novel method to enable an analytical stochastic solution of a  class of stochastic LWR models with both stochastic initial conditions and stochastic fundamental diagram. \citet{Jabari2012,Jabari2013} have considered the source of randomness in the LWR model by the uncertainty inherent in a driver gap choice, which is represented by random state dependent vehicle time headway. In this model, the problem of negative sample paths of the stochastic variables is well tackled.  

Microscopic models, on the other hand, describe the dynamics of traffic flow at a high level of the detail, such as the movement of individual vehicles longitudinally and laterally, i.e. car-following models and lane-changing models, respectively. Car-following models have been used widely to study the reaction of the driver with respect to his/her neighboring vehicles including the individual speed and acceleration/deceleration (see details in the book of \citet{Treiber2013}). In this microscopic approach, \citet{Laval2014} proposed a stochastic desired acceleration model which was extended from the car-following of \citet{Newell2002} by adding white noise to the driver's desired acceleration. It has been reported that this model can replicate traffic oscillations. \citet{Treiber2006} argued that the desired time gap in a car-following model (i.e. the Intelligent Driver Model-IDM) is traffic state dependent, e.g.  it is increased after a long time in congested traffic. The authors have modeled the desired time gap as a dynamical function of the speed variance and then replaced it in the IDM to replicate many interesting traffic phenomena: widely scattered flows, capacity drop, etc. Based on Newell$'$s car-following model \citep{Newell2002} , \citet{Jabari2014} have derived a (first order) macroscopic traffic model with a probabilistic fundamental diagram. A Lagrangian version of such model has been proposed in \citet{Jabari2018stochastic}, which represent the uncertain choice of the follower's free-speed, reaction times and safe distance with respect to the leader. An integrated recurrent neural network and the IDM has been proposed to predict traffic oscillations by \citet{Zhou2017}.  \citet{Tian2016a} have attempted to improve the IDM by considering the difference between the driving behaviors at high-speed and low-speed. More specifically, the desired time gap is defined as a discrete function as below: 
\begin{equation}
    T(t+\Delta t)=
\begin{cases}
    T_1+rT_2,& \text{with probability } p\\
    T(t)& \text{otherwise } 
\end{cases}
\label{randomgap}
\end{equation}
where $r$ is the uniformly distributed random number between 0 and 1. $T_1$ and $T_2$ are constants indicating the range of time gap variations, typically $T_1<T_2$. $\Delta t$ is a simulation time step. The authors have shown that the proposed discrete IDM with the randomized desired time gap in equation (\ref{randomgap})  can replicate the synchronized traffic flow patterns. Latter on this improved model has been used by \citet{Tian2016b} to simulate a concave growth of traffic oscillations. 

In general, it has been shown that allowing the desired time-headway to change over time (e.g. \citet{Tian2016a}) or adding noise to the continuous driver's desired acceleration (e.g. \citet{Laval2014}) have led to better model prediction. Along with this line, we aim to extend the model of \citet{Laval2014} to the form of coupled continuous stochastic differential equations which essentially captures the time-varying random choice of the driver's acceleration. The continuous form of the proposed model allows a theoretical insight into the effect of stochasticity on the stability of the traffic flow. This would help reveal quantitative relationships between car-following model structures and traffic oscillations due to uncertainties and stochasticity that are inherent from human car-following behavior. To the best of our knowledge, the stochastic stability analysis of a car-following model has only been conducted by \citet{Treiber2017}, which only focuses on two consecutive vehicles of a specific car-following law adapted from an acceleration-based model and only provides analytical results for a sub-critical regime. It is found that stochasticity, in fact, adds nothing different to linear stability. However, a more general method and theoretical insights into the instabilities through a stream of vehicles (e.g. string stability) are still missing for general stochastic car-following models.  Our approach advances the model of \citet{Laval2014} and \citet{Treiber2017} as following: 
\begin{itemize}
    \item We apply coupled Langevin equations to model the complex human driver behavior in aggregate lane car-following models which have a tractable mathematical foundation. In our model, an extended Cox-Ingersoll-Ross (CIR) stochastic process will be used to describe the stochastic human perception.  An important property of the extended CIR stochastic process is to enhance the non-negative properties of the stochastic traffic variables (e.g. individual-speed) for any arbitrary model parameters, which is not always the case in the model of \citet{Laval2014,Yuan2018}. 
    \item We derive general (string) stochastic linear stability conditions using the proposed model. These conditions are a first attempt to show analytically how the stochasticity affects the instabilities of traffic flow. More specifically, the derived conditions are able to describe the speed variations of the followers according to certain random human behaviors given the leading vehicle's speed. This paper thus fills the methodological gap of linear stability analysis for stochastic car-following models. The results complement the traditional stability conditions for deterministic car-following models, which have been used widely in the literature. More especially, our analytical results conform to the empirical results of the traffic instability in \citet{Jiang2018}. Our model shows that stochasticity does affect the stability of traffic flow at low-speed whereas the deterministic counterpart of the model shows a stable condition (i.e. Figures \ref{fig:sovm:stab:1} - \ref{Fig:traj}).
    \item The proposed model is verified against real-life traffic data (e.g. NGSIM trajectory data). The calibrated results show a good model performance. 
\end{itemize}


The rest of this paper is organized as follows. Section \ref{model} describes the model formulation and discusses some important properties of the proposed model. We derive the stochastic linear stability conditions of the proposed model in Section \ref{stabcondition}. Section \ref{Sec.Simulation} describes some numerical results supporting the stochastic stability conditions found in Section \ref{stabcondition} and verifies the performance of the proposed model using NGSIM data. Finally, we conclude the paper in Section \ref{Sec.Conclusion}.
\section*{Notation}
\begin{tabular}{ll}
  \textbf{Index}  &  \\
  $t$   & Time instant ($s$) \\
  \textbf{Model variables}  & \\
  $x_n$ & Location of vehicle $n$ ($m$) \\
  $v_n$ & Speed of vehicle $n$ ($m/s$) \\
  $\Delta v_n$ & Relative speed of vehicle $n$ and its leader ($n-1$) ($m/s$) \\
  $s_n$ & Bumper-to-bumper space gap between vehicle $n$ and its leader ($n-1$) ($m$)\\
  \textbf{Model parameters} & \\
  $v_0$ & Desired speed ($m/s$) \\
  $s_c$ & Critical headway ($m$)\\
  $\alpha$ & Dimensionless constant coefficient\\
  $\beta$ & Reaction coefficient ($1/s$)\\
  $\sigma_0$ & Dissipation coefficient ($\sqrt{m}/s$ for our model)
\end{tabular}

\section{Model formulation\label{model}}
The use of multiplicative Gaussian white noise to describe the the acceleration/deceleration of the follower with respect to the leader in a car-following model has been justified in \citet{Laval2014}. This was based on authors' collected data  containing the position, speed, acceleration and altitude information. The data was collected during the acceleration process at signalized intersections, and corresponds to instances where the vehicle was the leader of the platoon.

This paper adopts the Langevin equations to describe the stochastic driving behavior of the car-following models. In principle, the Langevin equations are used to illustrate the stochastic process in physics, which is an effective method for the modelling of quasi-continuous diffusion processes \citep{Manke2009}. This formalism has been used in a wide range of problems of Brownian motions, economics and financial mathematics, chemical reactions, and diverse optimization in the last decades. Moreover, a random characteristic of the variance in the modelling of stochastic dynamics has been actively used in financial modelling. This provides the possibility of considering efficiently unpredictable and uncontrolled effects of exogenous variables, called as a stochastic volatility.

To apply such Langevin method to our traffic problems, a generic stochastic equation for the acceleration of a vehicle $n$ includes a deterministic part and a Langevin source.
\begin{equation}
dv_n(t)=f\left(v_n(t), s_n(t),\Delta v_n(t)\right)dt + g\left(v_n(t), s_n(t),\eta(t)\right)dt\label{stoc}
\end{equation}
where $f(.)$ denotes a deterministic function depending on the state
of the considered vehicle and the leading vehicle while $g(.)$ denotes
the stochastic source which depends on only the state of the
considered vehicle and the stochastic process $\eta(t)$. When $g(.)=0$
we obtain a deterministic car-following model governed by a specific
definition of the function $f(.)$ such as the OVM \citep{Bando95},
FVDM \citep{Jiang2002} or IDM \citep{Treiber2005}. Without loss of
generality, in this paper, we use the OVM as a specific definition of
the function $f(.)$. In fact,  we have chosen the OVM as specific
underlying model since it is best suited for analytic
investigations. The actual contribution is the stochastic part which
can be applied to a wide class of deterministic car-following models,
e.g., for the IDM which will not lead to crashes, even in the linearly
unstable regime. To avoid a negative gaps in the numerical
implementation due to the choice of OVM, we set any negative gap to a
small positive value. The function $f(.)$ thus reads: 

\begin{equation}
f\left(v_n(t),s_n(t),\Delta v_n(t)\right)=\beta\left[V_{op}(s_n(t))-v_n(t)\right]\label{OVM}
\end{equation}
where $V_{op}(.)$ denotes the headway-dependent optimal speed. In this paper we adopt the following functional optimal speed:
\begin{equation}
V_{op}(s)=\frac{v_{0}}{2}\left[\tanh{\left(\frac{s}{s_c}-\alpha\right)}+\tanh{\alpha}\right]\end{equation}

To take into account the stochastic part $g(.)$, the multiplicative Gaussian white noise is used so that the model equation (\ref{stoc}) reads:
\begin{equation}
dv_n(t)=\beta\left[V_{op}(s_n(t))-v_n(t)\right]dt + \sigma\left(v_n(t)\right)dW_n(t)\label{SOVM}
\end{equation}
where $W_n(t)$ follows a Wiener process modelling the random deviations from the mean speed of the individual vehicle. $\sigma(v_n)$ is a positive speed dependent dissipation parameter describing the noise strength.  Higher $\sigma(v_n)$ implies more randomness in the acceleration of the follower.


In \citet{Laval2014}, $V_{op}(s_n(t))=v_c$ and $\sigma\left(v_n\right)=\sigma_0$ where $v_c$ is a target speed (the constant value measured from data) and $\sigma_0$ is a constant. This linear model is essentially an Ornstein-Uhlenbeck (OU) process \citep{Uhlenbeck1930} in which the solution for $v_n(t)$ (can be found in any stochastic differential equations textbook, e.g. \citet{Mao2008,Gardiner2009,Evans2014}, 
has the Gaussian distribution with asymptotic mean and variance:
\begin{eqnarray}
E\left[v_n(t)\right] &=& v_0e^{-\beta t} + v_c\left(1-e^{-\beta t}\right)\\
Var\left[v_n(t)\right] &=& \frac{\sigma_0^2}{2\beta}\left(1-e^{-2\beta t}\right)
\end{eqnarray}
It is obvious that at the equilibrium state:
\begin{equation}
\lim\limits_{t\rightarrow \infty} E\left[v_n(t)\right] = v_c;\hspace{0.5cm}\lim\limits_{t\rightarrow \infty} Var\left[v_n(t)\right] = \frac{\sigma_0^2}{2\beta}
\end{equation}
Therefore, at the equilibrium state the speed variance depends only on the model parameters (i.e.$\sigma_0$ and $\beta$).

The model is simply linear but can replicate well the data observations as described in \citet{Laval2014}. However, a problem in the above model type is that when the speed is low or the constant value of $\sigma_0$ is high (above the pre-determined upper bound) we could obtain negative values of speed. This is not avoidable due to the nature of the model although ones can set the speed to zero in the numerical implementation if it is negative.  

A revised version of   \citet{Laval2014}, developed by \citet{Yuan2018}, accounts for the speed-dependence of the acceleration variance:
\begin{equation}
dv_n(t)=\beta\left(v_c-v_n(t)\right)dt + \sigma_0\left(v_c-v_n(t)\right)dW_n(t)\label{SOVM_rev}
\end{equation}
The solution for $v_n(t)$ of this model follows log-normal distribution with  asymptotic mean and variance:
\begin{eqnarray}
E\left[v_n(t)\right] &=& v_c -  \left(v_c-v_0\right)e^{-\beta t}\\
Var\left[v_n(t)\right] &=& \left(v_c-v_0\right)^2 \left(e^{-(2\beta-\sigma_0^2) t}-e^{-2\beta t}\right)
\end{eqnarray}
Nevertheless, in the above models, the optimal speed of the deterministic part (i.e. drift term) is fixed to a constant or free-speed value, which limits the model performance in many cases. More specially, the above models cannot show how the stochasticity affects the instabilities in a range of traffic situations. To this end,  we propose an extended and more generalized model which:  
\begin{itemize}
    \item relaxes the assumption of the constant dissipation parameter,
    \item uses a state-dependent optimal speed, 
    \item enhances the non-negative properties of the speed for arbitrary values of the model parameters, and
    \item allows the derivation of generic string stochastic stability conditions of traffic flow. 
\end{itemize}
 To this end, we propose an extended Cox-Ingersoll-Ross (CIR) process \citep{Cox1985} to model the acceleration deviations:
\begin{equation}
dv_n(t)=\beta\left[V_{op}(s_n(t))-v_n(t)\right]dt + \sigma_0\sqrt{v_n(t)}dW_n(t)\label{SOVM2}
\end{equation}

An important and distinct property of the extended CIR process over the OU process lies in the variable-dependent standard deviation factor $\sigma_0\sqrt{v_n(t)}$, which attempts to avoid the negative trajectories of the stochastic variable $v_n(t)$ for any arbitrary values of  $\sigma_0$. More generally, when $v_n(t)$ is close to zero, the standard deviation $\sigma_0\sqrt{v_n(t)}$ becomes very small even for high values of $\sigma_0$, which reduces the effect of the random oscillation on the speed. Consequently, when the speed is close to zero, its evolution becomes dominated by the drift term (i.e. the deterministic part), which pushes the speed towards the (positive) value $V_{op}$. It is worth mentioning here that in the OVM, when traffic becomes (linearly) unstable, the speed can even be negative if the optimal speed function $V_{op}$ is not positively defined. To this end, we impose the following constraint in the numerical implementation: $V_{op}^{+}=max(0,V_{op})$. This will guarantee that the stochastic speed is pushed towards the non-negative value in the low speed regime. In the following sections, we will show that the effect of the stochastic part (i.e. the standard deviation of the acceleration) on traffic stability vanishes in the free-flow traffic regime.

\begin{REM}
If we set $V_{op}(s_n(t))=v_c$ as in the case of the model of \citet{Laval2014}, equation (\ref{SOVM2}) follows a standard CIR process, from which the solution for $v_n(t)$ has the Gaussian distribution with asymptotic mean and variance: 
\begin{eqnarray}
E\left[v_n(t)\right] &=& v_0e^{-\beta t} + v_c\left(1-e^{-\beta t}\right)\\
Var\left[v_n(t)\right] &=& \frac{v_0\sigma_0^2}{\beta}\left(e^{-\beta t}-e^{-2\beta t}\right)+\frac{v_c\sigma_0^2}{2\beta}\left(1-e^{-\beta t}\right)^2
\end{eqnarray}
and at the equilibrium state:
\begin{equation}
\lim\limits_{t\rightarrow \infty} E\left[v_n(t)\right] = v_c;\hspace{0.5cm}\lim\limits_{t\rightarrow \infty} Var\left[v_n(t)\right] = \frac{v_c\sigma_0^2}{2\beta}
\end{equation}
\end{REM}
In the CIR process, at the equilibrium state, unlike the model of \citet{Laval2014}, the speed variance also depends on the critical speed $v_c$. Note that, due to the contribution of $\sqrt{v_n(t)}$ in the stochastic part of the CIR-like model, the unit of $\sigma_0$ is different between our model and the model of \citet{Laval2014}.

\section{Linear stochastic stability\label{stabcondition}}
In order to understand the stability of the model in response to the multiplicative white noise (i.e. the Brownian motions), let us adopt the linear analysis method. In principle, the linear stability method has been used widely in traffic flow literature to derive the conditions influencing
the long-wavelength instabilities of traffic flow. We will adopt such method to our stochastic car-following model equation (\ref{SOVM2}). The stochastic model of \citet{Treiber2017} shows that stochasticity does not affect the linear stability of traffic flow. Whereas empirical results in \citet{Jiang2018} indicated that the stochastic nature of drivers does destabilize traffic flow at low speed, where the deterministic model shows a stable traffic pattern. This section describes a first attempt to derive the general linear stochastic stability condition of stochastic car-following models, which is consistent with the empirical findings in \citet{Jiang2018}.

More specifically, in the stationary situation, the speed of the
considered vehicle is given by $v_n=v_{n-1}=...=v_e$ and the gap is
given by $s_n=s_{n-1}=...=s_e$ which is calculated from a mean fundamental diagram $v_e=V_{op}(s_e)$. 
Now let $\delta v_n$ and  $\delta s_n$ denote the small deviation of
the speed and gap around the stationary situation: $v_{n} =
v_e+\delta v_{n}$ and $s_{n} = s_e+\delta s_{n}$.  First order Taylor
expansion of equation (\ref{SOVM2}) leads to a linear stochastically perturbed evolution equation:
\begin{equation}
d \delta v_n(t)=-\beta\delta v_n(t) dt+  \beta V_{op}^{'}\delta s_n(t) dt
+ \frac{\sigma_0}{2\sqrt{v_e}} \delta v_n(t) dW(t),\label{eq_microequil}
\end{equation}
Note, by definition, we have:
\begin{equation} 
d \delta s_n(t)=\left(\delta v_{n-1}(t)-\delta v_{n}(t)\right)dt\label{eq_microequil2}
\end{equation}

\subsection{Local stability condition}
In this section, we will study the evolution of the gap and speed deviation, i.e. $\delta s_n$ and $\delta v_n$, given model parameters and dissipation parameters via a local stability analysis. In principle, the local stability describes how a driver recovers from a small disturbance in real-life traffic and returns to his/her steady state over time. That is the local stability is guaranteed if the gap and speed fluctuations of the followers decrease over time, or at least, do not amplify. To derive the local stability conditions, we adopt the Lyapunov stability theory, which makes use of a Lyapunov function $V(\vect{U})$. If we can find a non-negative function $V(\vect{U})$ that always decreases along trajectories of the system, a locally stable equilibrium state can be found. 

Note for the local stability analysis, the leader's speed is fixed, and generally taken as a constant value so we set $\delta v_{n-1}=0$, and we consider the stability of the follower w.r.t. the leading vehicle behavior (that is why it is called local stability). Thus mathematically system of equations (\ref{eq_microequil}) and (\ref{eq_microequil2}) can be written as the following linear stochastic ordinary differential equations (SODEs):\\

\begin{equation}
d\vect{U}(t)=\vect{F}\vect{U}(t)dt + \vect{G}\vect{U}(t)d\vect{W}(t)\label{Slinear1}
\end{equation}
where
\begin{equation} 
\vect{U} =
\kbordermatrix{
	& \\
	& \delta s_n  \\
	& \delta v_n  \\
}
,\hspace{0.5cm}
\vect{F} =
\kbordermatrix{
	& &\\
	& 0 & -1  \\
	& \beta V_{op}^{'}& -\beta  \\
}
,\hspace{0.5cm}
\vect{G} =
\kbordermatrix{
	& &\\
	& 0 & 0  \\
	& 0 & \gamma  \\
}
\end{equation}
with $\gamma=\frac{\sigma_0}{2\sqrt{v_e}}$.

To follow the standard derivation of the Lyapunov stability condition in \citet{Mao2008}, we have the following definition:

\begin{DD}
Let $\Phi$ denote the family of all continuous non-decreasing functions $\phi$  such that $\phi(0)=0$ and $\phi(r)>0$ if $r>0$. The trivial solution of the stochastic system (\ref{Slinear1}) is said to be locally stochastic stable if there exist a Lyapunov function $V(\vect{U}(t))$ satisfying:
\begin{eqnarray}
        &&V(\vect{U}(t))\geq \phi(|\vect{U}(t)|), \hspace{.5cm} \phi \in\Phi  \hspace{.5cm}\text{(i.e. positive definite)}\label{eq_L2}\\
    &&LV(\vect{U}(t))=  V_1(\vect{U}(t))\vect{F}\vect{U}(t) +\frac{1}{2}\Tr\left[(G\vect{U}(t))^T V_2(\vect{U})G\vect{U}(t)\right]\leq 0 \hspace{.5cm}\text{(i.e. negative definite)}  \label{eq_L3}
\end{eqnarray}
where
\begin{eqnarray*}
V_1(\vect{U}) = \kbordermatrix{
	& \\
	& \frac{\partial V}{\partial\delta s_n}  \\
	& \frac{\partial V}{\partial\delta v_n}  \\
},\hspace{0.5cm} V_2(\vect{U}) = \kbordermatrix{
	& &\\
	& \frac{\partial ^2 V}{\partial\delta s_n^2} &  \frac{\partial^2 V}{\partial\delta s_n\partial\delta v_n}\\
	& \frac{\partial ^2V}{\partial\delta v_n\partial\delta s_n} &  \frac{\partial^2 V}{\partial\delta v_n^2}\\
}
\end{eqnarray*}
\end{DD}



Then we have the following theorem:
\begin{TT}
The stochastic car-following model (\ref{SOVM2}) is locally  stable if :
\begin{equation}
    \sigma_0^2 \leq 8\beta v_e
\end{equation}

\end{TT}
\begin{proof}
For the system (\ref{Slinear1}), we define a Lyapunov function as:
\begin{equation}
V(\vect{U}(t)) =\vect{U}^T(t) \vect{Q} \vect{U}(t)     
\end{equation}
with 
\begin{equation}
\vect{Q}=
\kbordermatrix{
	& &\\
	& \beta V_{op}^{'} & 0  \\
	& 0 & 1  \\  
}	
\end{equation}
which leads to $V(\vect{U}(t)) = \beta V_{op}^{'}\delta s_n^2(t) +\delta v_n^2(t)$ and 
\begin{equation}
LV(\vect{U}(t)) = \left(-2\beta + \frac{\sigma_0^2}{4v_e}\right)\delta v_n^2(t)  
\end{equation}
It is clear that, $V(\vect{U}(t))\geq \lambda_{min}(\vect{Q})\parallel U(t)\parallel ^2$  where $\lambda_{min}(\vect{Q})$ denotes the smallest eigvenvalues of the matrix $\vect{Q}$. Therefore,  condition (\ref{eq_L2}) is satisfied.  To guarantee $LV(\vect{U}(t))\leq 0$, i.e. condition (\ref{eq_L3}) , we need:  $\sigma_0^2 \leq 8\beta v_e$
\end{proof}

When $\sigma_0 = 0$ (i.e. the original OVM), the system is locally stable. However, large noise clearly destabilizes the local stability of the system. 

\subsection{String stability condition}
The string stability considers how a small perturbation in the gap and speed of the leader affect the gap and speed of all the followers. The string stability of deterministic car-following models has been studied extensively in the literature ( see example, \citet{Ngoduy2013} and references there-in). The contribution of this section is to derive the linear stability condition taking into account the multiplicative white noisy terms. Let $\delta s_n=\mathcal{S}e^{i\omega n+ \lambda t}$ and $\delta v_n=\mathcal{M}e^{i\omega n+ \lambda t}$, where $\mathcal{S}$ and $\mathcal{M}$ are constants, then the system of equations (\ref{eq_microequil}) and (\ref{eq_microequil2}) is also written as the following linear SODEs:
\begin{equation}
d\vect{U}(t)=\vect{F}\vect{U}(t)dt + \vect{G}\vect{U}(t)d\vect{W}(t)\label{Slinear2}
\end{equation}
where $\vect{U}(t)$  and $\vect{G}$ are defined as above, while 
\begin{equation} 
\vect{F} =
\kbordermatrix{
	& &\\
	& 0 & e^{-i\omega}-1  \\
	& \beta V_{op}^{'}& -\beta  \\
}\label{stringF}
\end{equation}
In the literature of stochastic process, there is a wide spectrum of stochastic stability analysis for both continuous and discrete stochastic differential equations \citep{kushner1971}. In this paper we expand these ideas to our proposed stochastic car-following model (\ref{SOVM2}) while focusing on almost sure stability, moment (i.e. 2nd order) stability (which is also called mean square stability), and the relationship between them. To follow the text in Evans (2014), the following definitions are used for the stochastic stability analysis:
\begin{DD}
	The stochastic system is said to be almost surely linearly stochastically stable (i.e. asymptotically stochastically stable in the large) if  the solution of equation  (\ref{Slinear2}) has the following convergence property:
	\begin{equation}
	\lim\limits_{t\rightarrow \infty} \parallel \vect{U}(t)\parallel=0 \text{ with probability 1}\label{stab1}
	\end{equation}
\end{DD}
It has been shown in literature of SDEs (e.g. \citet{Mao2008,Gardiner2009,Evans2014}) that condition (\ref{stab1}) is equivalent to:
\begin{equation}
\mathfrak{R}\{\vect{F^*}\} = \mathfrak{R}\{\vect{F}-0.5\parallel\vect{G}\parallel^2\}\leq 0  \label{stoccond1}
\end{equation}
where $\mathfrak{R}\{\vect{F^*}\}$ denotes the real parts of the eigenvalues of the Jacobian matrix $\vect{F^*}=\vect{F}-0.5\parallel\vect{G}\parallel^2$, with $\vect{F}$ being defined in (\ref{stringF}). 
\begin{DD}
	The stochastic system is said to be mean square stable (i.e. asymptotically stochastically mean square stable) if  the solution of equation  (\ref{Slinear2}) has the following convergence property:
	\begin{equation}
	\lim\limits_{t\rightarrow \infty} E\left[\parallel \vect{U}(t)\parallel^2\right]=0 \label{stab2}
	\end{equation}
\end{DD}

To follow the literature of SDEs (e.g. \citet{Mao2008,Gardiner2009,Evans2014}),  condition (\ref{stab2}) is equivalent to:
\begin{equation}
\mathfrak{R}\{\vect{F}\}+0.5\parallel\vect{G}\parallel^2\leq 0 \label{stoccond2}
\end{equation}
Note that if $\vect{G}=0$ conditions (\ref{stoccond1}) and (\ref{stoccond2}) are equivalent to the linear stability condition of the deterministic OVM, i.e. $\mathfrak{R}\{\vect{F}\}\leq 0$.
Then we have the following theorem:
\begin{TT}
The stochastic car-following model (\ref{SOVM2}) is almost sure stable if:
\begin{equation}
 \sigma_0^2\leq 8v_e \left(\beta- \sqrt{2\beta V_{op}^{'}}\right)\label{cond1}
\end{equation}
Obviously, condition (\ref{cond1}) is stricter than the local stability condition.
\end{TT}
\begin{proof}
It is straightforward to show that the eigenvalues of the Jacobian matrix $\vect{F^{*}}$ is the solutions of the following characteristic equation:
\begin{equation}
\lambda^2+\left(\beta-\frac{\sigma_0^2}{8v_e}\right)\lambda-\beta V_{op}^{'}\left(e^{-i\omega}-1\right)=0\label{eq_lineareqn1}
\end{equation}

Let us expand $\lambda$ in a power series
solution $\lambda=i\omega\lambda_1+\omega^2\lambda_2+...$ where
$\lambda_1$ and $\lambda_2$ are real coefficients.  Let us
substitute this expansion into equation (\ref{eq_lineareqn1}) and expand the exponential terms to the second order, and set the
first order ($\mathcal{O}(\omega)$) and the second order
($\mathcal{O}(\omega^2)$) terms to zero. After a rather lengthy but
straightforward algebraic calculation we have the following solutions:

\begin{equation}
\lambda_1=-\frac{\beta V_{op}^{'}}{\beta^*},
\end{equation}
and
\begin{equation}
\lambda_2=\frac{1}{\beta^{*}}\left(\lambda_1^2-\frac{\beta V_{op}^{'}}{2}\right)
\end{equation}
where $\beta^{*}=\beta-\frac{\sigma_0^2}{8v_e}$.\\
Condition (\ref{stoccond1}) holds if $\lambda_1 \leq 0$ and  $\lambda_2 \leq 0$. It is obvious that  $\lambda_1 \leq 0$ leads to $\beta^{*}\leq 0$. This is equivalent to $\sigma_0^2\leq 8\beta v_e$.  Whereas $\lambda_2 \leq 0$ leads to:
\begin{equation}
    \sigma_0^2\leq 8v_e \left(\beta- \sqrt{2\beta V_{op}^{'}}\right) 
\end{equation}
which also satisfies: $\sigma_0^2\leq 8\beta v_e$.
\end{proof}

\begin{TT}
The stochastic car-following model (\ref{SOVM2}) is mean square stable if:
\begin{equation}
 \sigma_0^2\leq \frac{4v_eV_{op}^{'}}{\beta} \left(\beta-2V_{op}^{'}\right)\label{cond2}
\end{equation}
\end{TT}
\begin{proof}
Following the proof in Theorem 2, the eigenvalues of the Jacobian matrix $\vect{F}$ is the solutions of the following characteristic equation:
\begin{equation}
\lambda^2+\beta\lambda-\beta V_{op}^{'}\left(e^{-i\omega}-1\right)=0\label{eq_lineareqn2}
\end{equation}
which lead to the following solutions:
\begin{equation}
\lambda_1=-V_{op}^{'}
\end{equation}
and
\begin{equation}
\lambda_2=\frac{1}{\beta}\left(\lambda_1^2-\frac{\beta V_{op}^{'}}{2}\right)
\end{equation}
Since $\lambda_1 \leq 0$, condition (\ref{stoccond2}) holds if and only if  $\lambda_2 + 0.5\gamma^2\leq 0$ which is equivalent to:
\begin{equation}
 \sigma_0^2\leq\frac{4v_eV_{op}^{'}}{\beta} \left(\beta-2V_{op}^{'}\right)
\end{equation}
\end{proof}
\begin{REM}
When $\sigma_0=0$, both conditions (\ref{stoccond1}) and (\ref{stoccond2}) become $\beta\geq 2 V_{op}^{'}$, which is the stability condition of the original OVM.
\end{REM}
\begin{REM}
The almost sure stability condition also implies the local stability condition. This is true for both determinitic and stochastic OVM.
\end{REM}
\begin{REM}
For the stochastic car-following model (\ref{SOVM2}), the mean square stability condition (i.e. $\mathfrak{R}\{\vect{F}\}+0.5\parallel\vect{G}\parallel^2\leq 0$) also implies the almost sure stability condition since $\mathfrak{R}\{2\vect{F}-\parallel\vect{G}\parallel^2\}\leq 2\mathfrak{R}\{\vect{F}\}+\parallel\vect{G}\parallel^2\leq 0$, which leads to  $\mathfrak{R}\{\vect{F}-0.5\parallel\vect{G}\parallel^2\}\leq 0$. 
 \end{REM}
\section{Model performance}
\label{Sec.Simulation}
The proposed model is numerically simulated  using a standard Euler-–Maruyama scheme. Details of the discretization is given in the Appendix. This section illustrates numerically the effects of the model parameters on the (linear) traffic stability and the performance of the proposed model in replicating real-world traffic data.

\subsection{Stochastic linear stability diagrams}
\begin{figure}[htp]
	\centering
\includegraphics[width=0.75\textwidth]{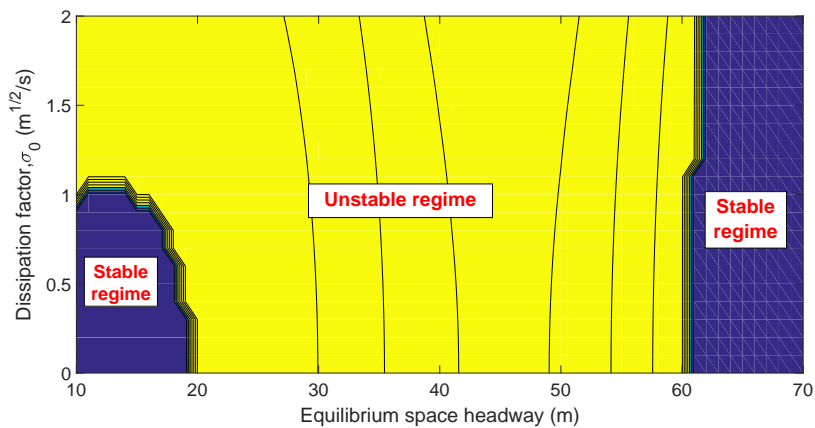}
	\caption{Stochastic almost sure stability diagram for a range of $\sigma_0$}
	\label{fig:sovm:stab:1}
\end{figure}
\begin{figure}[htp]
	\centering
\includegraphics[width=0.75\textwidth]{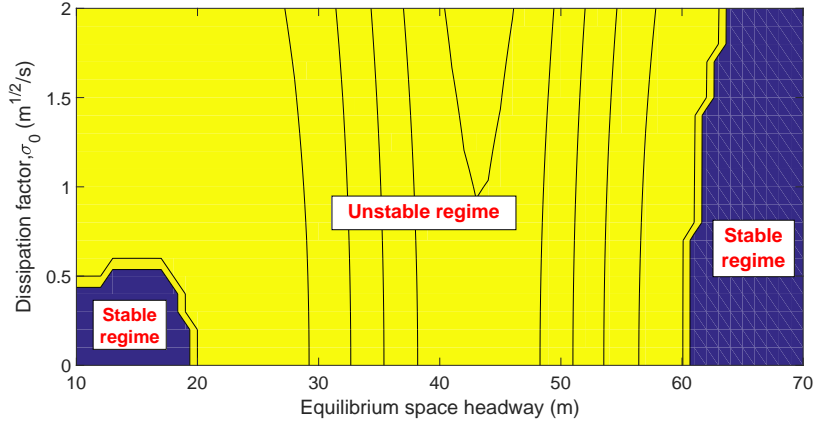}
	\caption{Stochastic mean square stability diagram for a range of $\sigma_0$}
	\label{fig:sovm:stab:2}
\end{figure}

Let us study the stability diagrams for various values of $\sigma_0$ against a range of the initial equilibrium space headway in Figures \ref{fig:sovm:stab:1} and \ref{fig:sovm:stab:2}. In these figures, the following model parameters are used: $\beta =0.5 s^{-1}$, $v_0=25$m/s, $s_c=20$m, and $\alpha=2$.  The unstable regime covers all the values of $\sigma_0$ and $s_e$ so that the stability conditions (\ref{cond1}) and (\ref{cond2}) are violated (i.e. yellow regime), respectively. Whereas, in the stable regime the stability conditions (\ref{cond1}) and (\ref{cond2}) hold (i.e. dark-blue regimes), respectively. It is seen in both figures that when traffic is in the free-flow stable regime (i.e. when the space-headway is relatively large), the stochastic factor has insignificant impact on the stability of traffic flow. However, when traffic is in the congested stable regime (i.e. the space headway is relatively small, $s_e\leq ~20$m), increasing stochastic factor tends to destabilize significantly traffic flow. This has been confirmed from empirical investigations in \citet{Jiang2018}: a small perturbation dies out in the deterministic model (i.e. traffic is stable), whereas the stochastic nature of traffic flow leads to the instability at low speed conditions. For example, in Figure \ref{fig:sovm:stab:1}, when $\sigma_0\geq 1 \sqrt{m}/s$ traffic is unstable for all space headway below 60m. Where as in Figure \ref{fig:sovm:stab:2} when $\sigma_0\geq 0.5\sqrt{m}/s$ traffic is unstable for all space headway below 60m. It is also clear that the stable regimes for almost sure stability condition are larger than those for mean square stability condition (i.e. this conforms to our Remark 3). 

\begin{figure}
    \begin{center}
        \begin{minipage}[b]{0.45\textwidth}
            \centering
            \subfigure[Stable traffic, replicated by OVM]{\label{Fig:stab}
            \includegraphics[width=9cm]{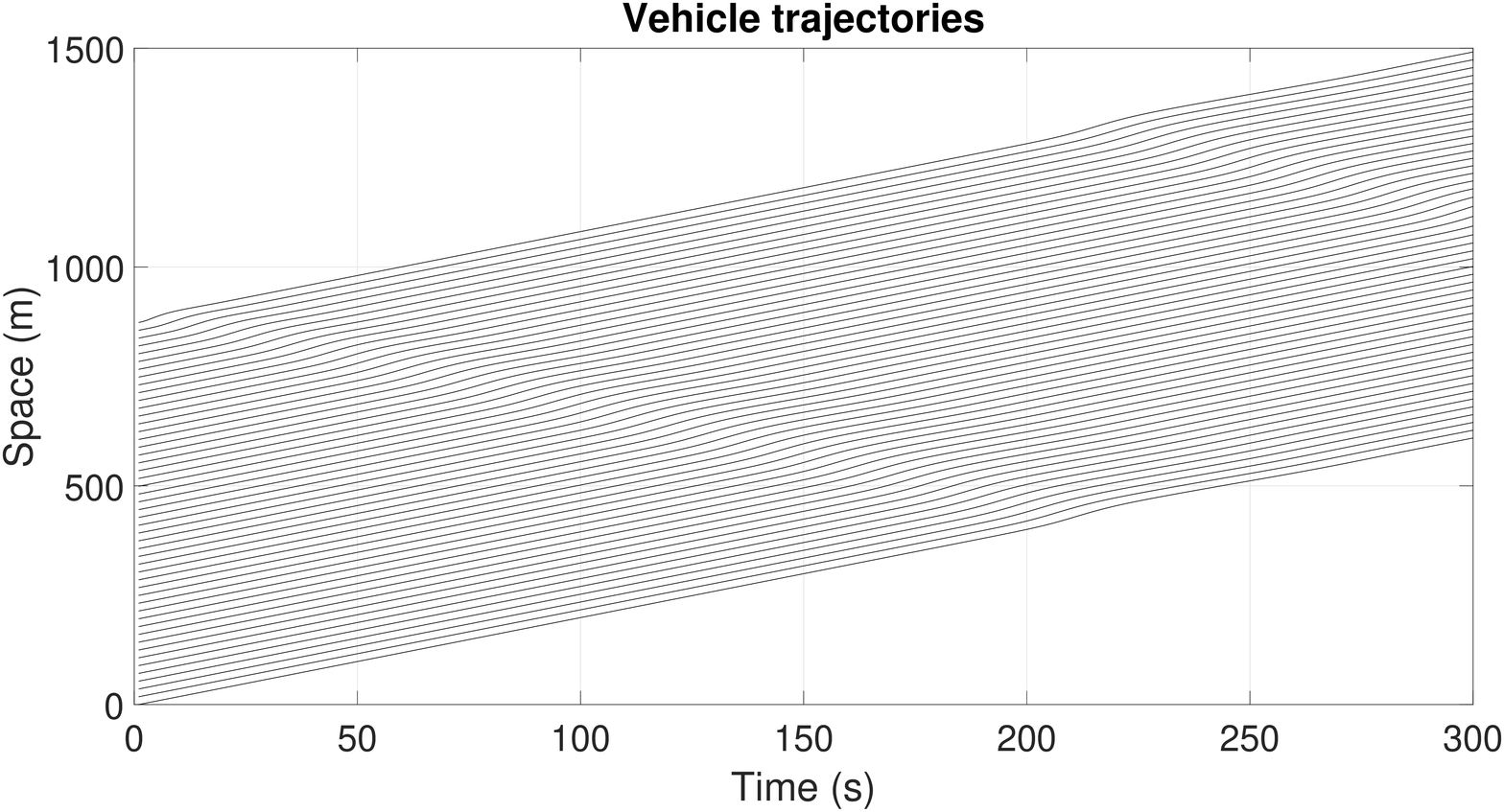}}
        \end{minipage}
        \begin{minipage}[b]{0.45\textwidth}
                \centering
                \subfigure[Unstable traffic, replicated by SOVM, $\sigma_0=1\sqrt{m}/s$]{\label{Fig:unstab}
                \includegraphics[width=9cm]{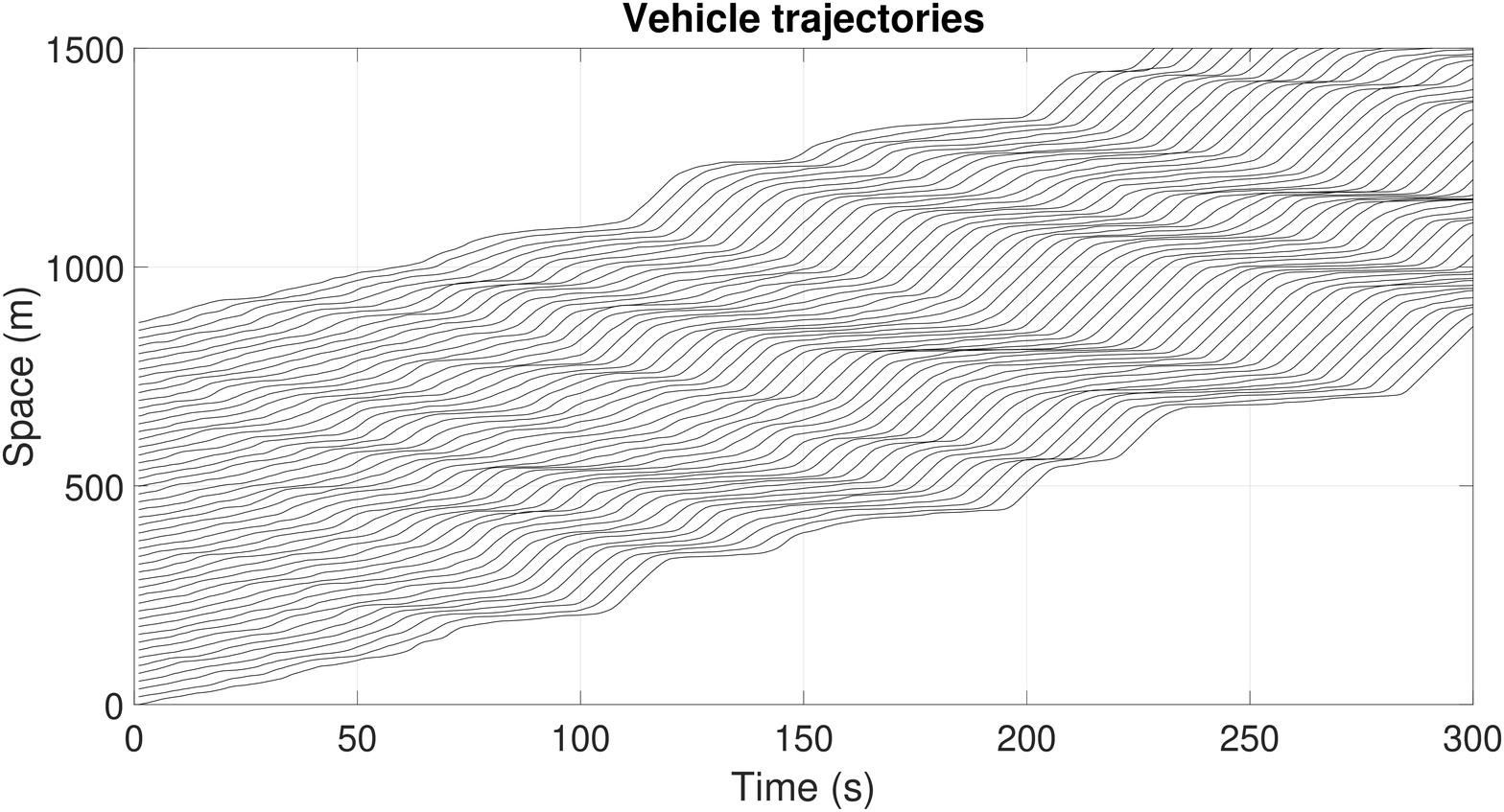}}
         \end{minipage}
    \end{center}
\bf\caption{Trajectories of 50 vehicles under the initial equilibrium space of 18m}\label{Fig:traj}
\end{figure}

Figure \ref{Fig:traj} shows the trajectories of 50 vehicles, replicated by the OVM (Fig. \ref{Fig:stab}) and the SOVM (Fig. \ref{Fig:unstab}) using the above model parameters and the initial equilibrium space headway $s_e=18$m. With these model settings, the stability condition of the OVM is $\beta-2V_{op}^{'}=0.05>0$ which indicates a stable traffic condition at the low equilibrium speed. In contrast, the right-hand side of the stability condition (\ref{cond2})  of the SOVM is $0.1872<\sigma_0^2=1$ indicating an unstable traffic condition at the low equilibrium speed.  


\subsection{Model calibration with real data}
In this section, we briefly describe the calibration results of the proposed model against the real data. More extensive model verification is given in a separate paper \citep{Lee2018}. The model parameters to be calibrated are:
\begin{itemize}
    \item Free-flow speed $v_0$  (m/s)
    \item Reaction coefficient $\beta$  (1/s)
    \item Critical headway $s_c$  (m)
    \item Constant coefficient $\alpha$  (dimensionless)
    \item Dissipation coefficient $\sigma_0$  ($\sqrt{m}/s$ )
\end{itemize}
Note that, similar to the model of \citet{Laval2014}, our proposed model only generates one additional parameter from the original OVM, that is the dissipation coefficient $\sigma_0$.
\begin{figure}[htp]
	\centering
\includegraphics[width=1\textwidth]{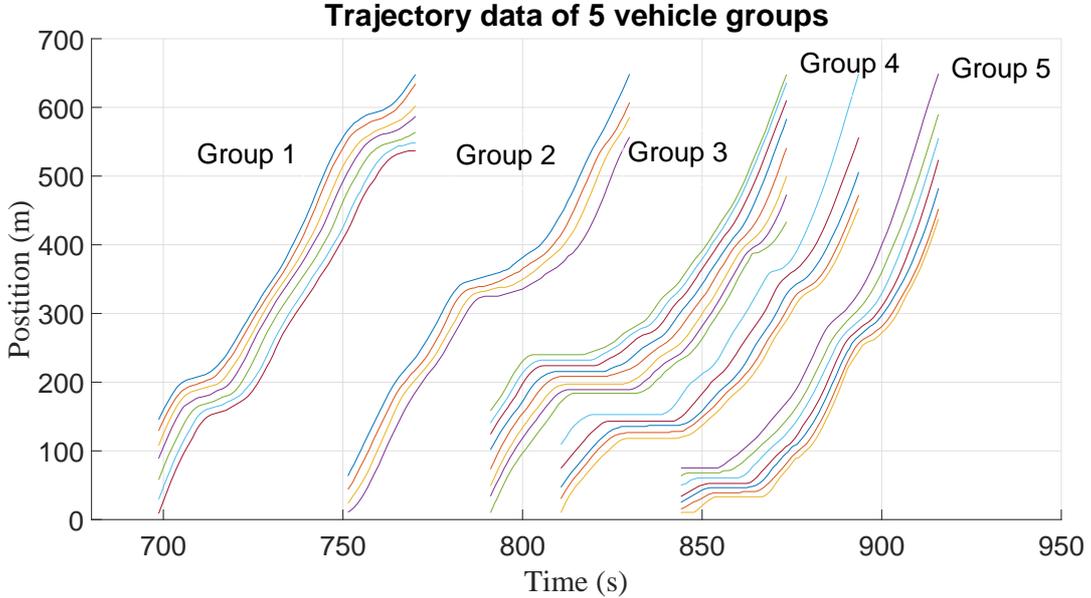}
	\caption{Trajectory data used for the model calibration}
	\label{fig:data}
\end{figure}

To verify the effectiveness of the proposed model, real vehicle trajectory data set collected in a U.S. freeway. The dataset and detailed information were provided by the Federal Highway Administration’s Next Generation Simulation (NGSIM). The trajectory data at different time periods are described in the Figure \ref{fig:data}. We select the trajectory data in the 3rd vehicle group (i.e. at 800s) where there was a substantial traffic congestion at position 200m, for our calibration. To calibrate the model, a standard genetic algorithm is used for our purposes where 100 replications will be used in the stochastic simulation. The mean of the individual vehicle speed is used to compare with the observed individual vehicle speed where the total mean squared errors between the model output and the data is used as a performance index (PI):
\begin{equation}
z=\sum_{j=1}^J\sqrt{\frac{1}{K}\sum_{k=1}^K\left[v_j(k|\vect\pi)-\tilde{v}_j(k)\right]^2}
\end{equation}
where $k\in K$ is time step, $v_j(k|\vect\pi)$ denotes the simulated speed of vehicle $j$ at time step $k$ given the model parameter $\vect\pi=[v_0\hspace{0.1cm} \beta \hspace{0.1cm} s_c\hspace{0.1cm} \alpha \hspace{0.1cm} \sigma_0]$. In contrast, $\tilde{v}_j(k)$ is the measured speed of vehicle $j$ at time step $k$. 



\begin{table}[htp]
\caption{Optimal model parameters obtained from the GA.}
\begin{center}
\begin{tabular}{ |c|c|c| c|c|} 
 \hline
 Free-flow speed & Reaction coefficient & Critical headway & Constant coefficient & Dissipation coefficient \\
 $v_0 $($m/s$) & $\beta$ ($1/s$) & $s_c$ ($m$) & $\alpha$ (dimensionless) & $\sigma_0$ ($\sqrt{m}/s$) \\ 
 \hline
 17.65 & 0.65 & 8.20 & 1.85  & 0.88 \\ 
 \hline
\end{tabular}
\label{tab:opt}
\end{center}
\end{table}
\begin{figure}[htp]
	\centering
\includegraphics[width=.75\textwidth]{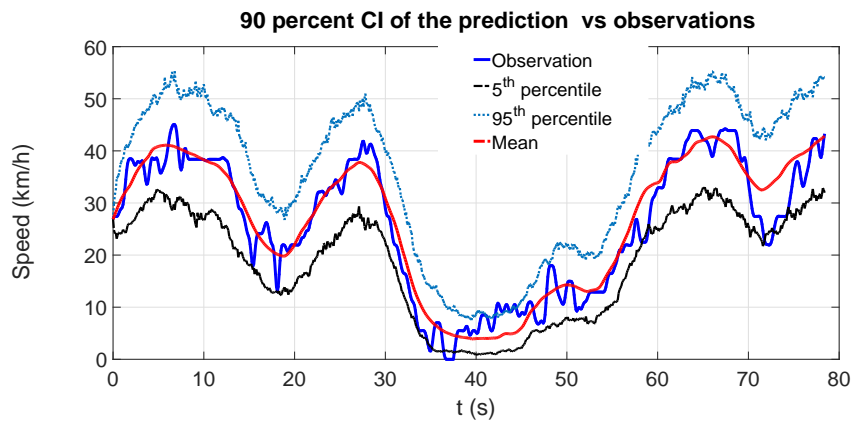}
	\caption{Model prediction vs observation}
	\label{fig:speed}
\end{figure}
\begin{figure}[htp]
	\centering
\includegraphics[width=.75\textwidth]{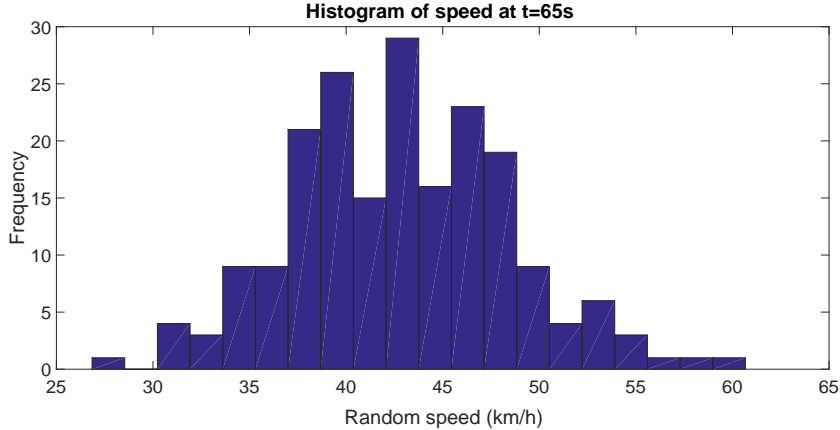}
	\caption{Distribution of the predicted speed at $t=65$s}
	\label{fig:final}
\end{figure}
Table \ref{tab:opt} shows the optimal model parameters obtained after 30 optimization iterations. Figure \ref{fig:speed} shows the predicted speed of the 2nd vehicle against the observed values. In Figure \ref{fig:speed}, the upper and lower bound show the trailing vehicle speed ($90\%$ confidence interval-CI). The solid blue line shows the observed speed while the dashed red line shows the mean of the predicted speed. Figure \ref{fig:speed} illustrates that the observed speed fits well to the range of the stochastic model output at $90\%$ CI and matches well to the mean of the model output. It is seen that the observed speed is always within the upper and lower bound of the predicted speed (at 90\% confidence interval). Moreover,  the variance of the predicted stochastic speed is relatively low when the speed is low (e.g. at time $t=40$s). A wide range of the predicted speeds at a certain time instant are described in Figure \ref{fig:final}, which conforms to a Gaussian distribution with the mean value being identified from Figure  \ref{fig:speed} (i.e. approximately 40\textit{km/h}).


\section{Concluding remarks}
\label{Sec.Conclusion}
The mathematical framework of stochastic car-following models developed in this study can deal with uncertain human perception. This is achieved by integrating the OVM in the stochastic equation. The proposed model is a first attempt to allow us to understand analytically how human errors can be responsible for traffic instabilities where the deterministic part is stable. This is achieved by relaxing the assumptions of constant dissipation parameter and constant optimal speed in the stochastic acceleration of \citet{Laval2014}. Moreover, the formulation of the proposed model follows an extended CIR stochastic process which consequently enhances non-negative speed values for arbitrary model parameters.  The model calibration results show good consistency with the trajectory data collected in a US freeway (i.e. NGSIM data). Our concurrent research is to extend the proposed method for multi-lane traffic dynamics, in which the lane-changing probability is continuously estimated using a deep learning method \citep{Lee2018}.

Finally, we would like to mention that Connected and Autonomous vehicles (CAVs) have been verified
to significantly improve traffic efficiency. However, there is still a long lifespan for heterogeneous traffic flow consisting of both human-driven vehicles and CAVs. Thus a deep understanding of the mixed traffic dynamics including both CAVs and human-driven vehicles is critical to the traffic stability issues for the deployment of CAVs in the near future. Many efforts have been made to study the impact of CAVs on traffic flow instability such as  \citet{Ngoduy2013a,Ngoduy2013,Talebpour2016,Wang2017}. Recently, \citet{Jia2019}  have proposed a novel model to consider the behavior of the CAVs in a heterogeneous platoon in detail but ignored the stochasticity of the drivers in the human-driven vehicles. In our future work, by introducing the proposed SOVM to model the stochastic human behavior, we hope to extend our approach to study the impact of CAVs on traffic instability under more realistic traffic flow situations.


\appendix
\section{Discrete-time stochastic model}

It is straightforward to formulate our model in the following form for numerical simulation.


\begin{equation}
d\vect{u}=\vect{A}\left(\vect{u},t\right)dt+\vect{B}\left(\vect{u},t\right)d\vect{W}(t)
\end{equation}
where
\begin{equation} \label{eq:A}
\vect{u}(t) =
\kbordermatrix{
	&\\
	& x_1(t)  \\
	& \vdots   \\
	& x_N(t)  \\
	& v_1(t)  \\
	& \vdots  \\
	& v_N(t) \\
},\hspace{0.5cm}
\vect{A}\left(\vect{u},t\right) =
\kbordermatrix{
	&\\
	& v_1(t)  \\
	& \vdots   \\
	& v_N(t)  \\
	& \beta\left(V_{op}(s_1(t))-v_1(t)\right)  \\
	& \vdots  \\
	& \beta\left(V_{op}(s_N(t))-v_N(t)\right) \\
}
\end{equation}

and
\begin{equation} 
\vect{B}\left(\vect{u},t\right) =
\kbordermatrix{
	& &\\
	& \vect{O}& \vect{O}  \\
	& \vect{O}& \vect{C}\left(\vect{u},t\right)  \\
}
,\hspace{0.5cm}
\vect{W}(t) =
\kbordermatrix{
	&\\
	& 0  \\
	& \vdots   \\
	& 0  \\
	& W_1(t)  \\
	& \vdots  \\
	& W_N(t) \\
}
\end{equation}
with
\begin{equation} 
\vect{C}\left(\vect{u},t\right) =
\kbordermatrix{
	& & & &\\
	& \sigma_0\sqrt{v_1(t)}& 0&\hdots & 0  \\
	& 0& \sigma_0\sqrt{v_2(t)} &\hdots & 0  \\
	& \vdots&  \vdots &\vdots &\vdots\\
	& 0& 0 &\hdots & \sigma_0\sqrt{v_N(t)}  \\
}\hspace{0.5cm}\text{for our proposed model}
\end{equation}
or
\begin{equation} 
\vect{C}\left(\vect{u},t\right) =
\kbordermatrix{
	& & & &\\
	& \sigma_0& 0&\hdots & 0  \\
	& 0& \sigma_0 &\hdots & 0  \\
	& \vdots&  \vdots &\vdots &\vdots\\
	& 0& 0 &\hdots & \sigma_0  \\
}\hspace{0.5cm}\text{for Laval et al.-type model}
\end{equation}

%
The models are simulated by using an explicit Euler-–Maruyama scheme. The discretization of the SOVM is
\begin{equation}
\vect{u}(k+1)-\vect{u}(k)=\vect{A}\left(\vect{u}(k)\right)\Delta t +\vect{B}\left(\vect{u}(k)\right)\sqrt{\Delta t}\vect{w}(k+1)
\end{equation}
\begin{equation} 
\vect{w}(k) =
\kbordermatrix{
	&\\
	& 0  \\
	& \vdots   \\
	& 0  \\
	& w_1(k)  \\
	& \vdots  \\
	& w_N(k) \\
},\hspace{0.5cm} w_n(k) \sim \mathcal{N}(0,1)
\end{equation}

If we change $\sigma_0$ by $\delta$, the change in the speed in a time step is:
\begin{eqnarray}
\Delta v_n(k+1)&=&\delta\sqrt{\Delta t}w_n(k+1),\hspace{0.2cm}\text{for model of \citet{Laval2014}}\\
\Delta v_n(k+1)&=&\delta\sqrt{v_n(k)}\sqrt{\Delta t}w_n(k+1),\hspace{0.2cm}\text{for our model}
\end{eqnarray}
It shows that the effect of a change in $\sigma_0$ in our model is insignificant when speed is low. It is also worth mentioning that the change is greatly affected by the random term.




\begin{thebibliography}{37}
\expandafter\ifx\csname natexlab\endcsname\relax\def\natexlab#1{#1}\fi
\expandafter\ifx\csname url\endcsname\relax
  \def\url#1{\texttt{#1}}\fi
\expandafter\ifx\csname urlprefix\endcsname\relax\def\urlprefix{URL }\fi

\bibitem[{Bando et~al.(1995)Bando, Hasebe, Nakayama, Shibata, and
  Sugiyama}]{Bando95}
Bando, M., Hasebe, K., Nakayama, A., Shibata, A., Sugiyama, Y., 1995.
  {Dynamical Model of Traffic Congestion and Numerical Simulation}. Physical
  Review E 51, 1035--1042.

\bibitem[{Cox et~al.(1985)Cox, Ingersoll, and Ross}]{Cox1985}
Cox, J.~C., Ingersoll, J.~E., Ross, S.~A., 1985. {A Theory of the Term
  Structure of Interest Rates}. Econometrica 53, 385--407.

\bibitem[{Evans(2014)}]{Evans2014}
Evans, L., 2014. {An introduction to stochastic differential equations.} {AMS},
  USA.

\bibitem[{Gardiner(2009)}]{Gardiner2009}
Gardiner, C., 2009. {Stochastic Methods: A Handbook for the Natural and Social
  Sciences.} {Spinger}, Germany.

\bibitem[{Jabari et~al.(2014)Jabari, Zheng, and Liu}]{Jabari2014}
Jabari, S., Zheng, J., Liu, H., 2014. {A probabilistic stationary
  speed–density relation based on Newell’s simplified car-following model}.
  {Transportation Research Part B} 68, 205--2234.

\bibitem[{Jabari and Liu(2012)}]{Jabari2012}
Jabari, S.~E., Liu, H.~X., 2012. {A stochastic model of traffic flow:
  Theoretical foundations}. {Transportation Research Part B} 46, 156--174.

\bibitem[{Jabari and Liu(2013)}]{Jabari2013}
Jabari, S.~E., Liu, H.~X., 2013. {A stochastic model of traffic flow: Gaussian
  approximation and estimation}. {Transportation Research Part B} 47, 15--41.

\bibitem[{Jabari et~al.(2018)Jabari, Zheng, Liu, and
  Filipovska}]{Jabari2018stochastic}
Jabari, S.~E., Zheng, F., Liu, H., Filipovska, M., 2018. Stochastic lagrangian
  modeling of traffic dynamics. In: The 97th Annual Meeting of the
  Transportation Research Board, Washington, DC. No. 18-04170.

\bibitem[{Jia et~al.(2019)Jia, Ngoduy, and Vu}]{Jia2019}
Jia, D., Ngoduy, D., Vu, H., 2019. {A multiclass microscopic model for
  heterogeneous platoon with vehicle-to-vehicle communication}.
  {Transportmetrica B} 7, 448--472.

\bibitem[{Jiang et~al.(2018)Jiang, Jin, Zhang, Huang, Tiang, Wang, M.B., Wang,
  and Jia}]{Jiang2018}
Jiang, R., Jin, C., Zhang, H., Huang, Y., Tiang, J., Wang, W., M.B., H., Wang,
  H., Jia, B., 2018. {Experimental and empirical investigations of traffic
  instability}. {Transportation Research Part C} 94, 83--98.

\bibitem[{Jiang et~al.(2002)Jiang, Wu, and Zhu}]{Jiang2002}
Jiang, R., Wu, Q., Zhu, Z., 2002. {Full Velocity Difference Model for a
  Car-Following Theory}. {Physical Review E } 64, {017101–--017104}.

\bibitem[{Kushner(1971)}]{kushner1971}
Kushner, H., 1971. {Introduction to Stochastic Control}. {Holt and Rinehart and
  Winston}.

\bibitem[{Laval and Chilukuri(2013)}]{Laval2013}
Laval, J.~A., Chilukuri, B.~R., 2013. {The Distribution of Congestion on a
  Class of Stochastic Kinematic Wave Models}. {Transportation Science} 48,
  217--224.

\bibitem[{Laval et~al.(2014)Laval, Toth, and Zhou}]{Laval2014}
Laval, J.~A., Toth, C.~S., Zhou, Y., 2014. {A parsimonious model for the
  formation of oscillations in car-following models}. {Transportation Research
  Part B} 70, 228--238.

\bibitem[{Lee et~al.(2019)Lee, Ngoduy, and Keyvan-Ekbatani}]{Lee2018}
Lee, S., Ngoduy, D., Keyvan-Ekbatani, M., 2019. {Integrated deep learning and
  stochastic car-following model for traffic dynamics on multi-lane freeways}.
  {Transportation Research Part C} {submitted}.

\bibitem[{Li et~al.(2012)Li, Chen, Wang, and Ni}]{Li2012}
Li, J., Chen, Q., Wang, H., Ni, D., 2012. {Analysis of LWR model with
  fundamental diagram subject to uncertainties}. {Transportmetrica A: Transport
  Science} 8, 387--405.

\bibitem[{Mahnke et~al.(2009)Mahnke, Kaupuzs, and Lubashevsky}]{Manke2009}
Mahnke, R., Kaupuzs, J., Lubashevsky, I., 2009. {Physics of stochastic
  processes: how randomness acts in time}. {John Wiley \& Sons}.

\bibitem[{Mao(2008)}]{Mao2008}
Mao, X., 2008. {Stochastic differential equations and applications.} {Horwood},
  Chichester.

\bibitem[{Newell(2002)}]{Newell2002}
Newell, F.~G., 2002. {A simplified car-following theory: a lower order model}.
  {Transportation Research Part B} 36, 195--205.

\bibitem[{Ngoduy(2011)}]{Ngoduy2011}
Ngoduy, D., 2011. {Multiclass first-order traffic model using stochastic
  fundamental diagrams}. Transportmetrica 7, 111--125.

\bibitem[{Ngoduy(2013{\natexlab{a}})}]{Ngoduy2013}
Ngoduy, D., 2013{\natexlab{a}}. {Analytical studies on the instabilities of
  heterogeneous intelligent traffic flow.} {Communications in Nonlinear Science
  and Numerical Simulation} 18, 2699--2706.

\bibitem[{Ngoduy(2013{\natexlab{b}})}]{Ngoduy2013a}
Ngoduy, D., 2013{\natexlab{b}}. {Instabilities of cooperative adaptive cruise
  control traffic flow: a macroscopic approach.} {Communications in Nonlinear
  Science and Numerical Simulation} 18, 2838--2851.

\bibitem[{Sumalee et~al.(2011)Sumalee, Zhong, Pan, and Szeto}]{Sumalee2011}
Sumalee, A., Zhong, R.~X., Pan, T.~L., Szeto, W.~Y., 2011. {Stochastic cell
  transmission model (SCTM): a stochastic dynamic traffic model for traffic
  state surveillance and assignment}. {Transportation Research Part B} 45,
  507--533.

\bibitem[{Talebpour and Mahmassani(2016)}]{Talebpour2016}
Talebpour, A., Mahmassani, H., 2016. {Influence of connected and autonomous
  vehicles on traffic flow stability and throughput}. {Transportation Research
  Part C} 71, 143--163.

\bibitem[{Tian et~al.(2016b)Tian, Jiang, Jia, Gao, and Ma}]{Tian2016b}
Tian, J., Jiang, R., Jia, B., Gao, Z., Ma, S., 2016b. {Empirical analysis and
  simulation of the concave growth pattern of traffic oscillations}.
  {Transportation Research Part B} 93, 338--354.

\bibitem[{Tian et~al.(2016a)Tian, Jiang, Li, Treiber, Jia, and Zhu}]{Tian2016a}
Tian, J., Jiang, R., Li, G., Treiber, M., Jia, B., Zhu, C., 2016a. {Improved 2D
  intelligent driver model in the framework of three-phase traffic theory
  simulating synchronized flow and concave growth pattern of traffic
  oscillations}. {Transportation Research Part F} 41, 55--65.

\bibitem[{Tordeux et~al.(2014)Tordeux, Roussignol, Lebacque, and
  Lassarre}]{Torduex2014}
Tordeux, A., Roussignol, M., Lebacque, J.~P., Lassarre, S., 2014. {A stochastic
  jump process applied to traffic flow modelling}. {Transportmetrica A:
  Transport Science} 10, 350--375.

\bibitem[{Treiber and Kesting(2013)}]{Treiber2013}
Treiber, M., Kesting, A., 2013. {Traffic Flow Dynamics}. {Springer}, Germany.

\bibitem[{Treiber and Kesting(2017)}]{Treiber2017}
Treiber, M., Kesting, A., 2017. {The Intelligent Driver Model with
  stochasticity - New insights into traffic flow oscillations}. {Transportation
  Research Part B} 23, 174--187.

\bibitem[{Treiber et~al.(2005)Treiber, Kesting, and Helbing}]{Treiber2005}
Treiber, M., Kesting, A., Helbing, D., 2005. {Delays, inaccuracies and
  anticipation in microscopic traffic model}. Physica A 360, 71--88.

\bibitem[{Treiber et~al.(2006)Treiber, Kesting, and Helbing}]{Treiber2006}
Treiber, M., Kesting, A., Helbing, D., 2006. {Understanding widely scattered
  traffic flows, the capacity drop, and platoon as effects of variance-driven
  time gaps}. Physical Review E 74, {0161231--0161239}.

\bibitem[{Uhlenbeck and Ornstein(1930)}]{Uhlenbeck1930}
Uhlenbeck, G.~E., Ornstein, L.~S., 1930. {On the theory of Brownian Motion}.
  Physical Review 36, 823--841.

\bibitem[{Wang et~al.(2017)Wang, Li, and Work}]{Wang2017}
Wang, R., Li, Y., Work, D., 2017. {Comparing traffic state estimators for mixed
  human and automated traffic flows}. {Transportation Research Part C} 78,
  95--110.

\bibitem[{Yeo and Skabardonis(2009)}]{Yeo2009}
Yeo, H., Skabardonis, A., 2009. {Understanding stop-and-go traffic in view of
  asymmetric traffic theory}. {International Symposium on Transportation and
  Traffic Theory }, {99--115}.

\bibitem[{Yuan et~al.(2018)Yuan, Laval, Knoop, Jiang, and
  Hoogendoorn}]{Yuan2018}
Yuan, K., Laval, J., Knoop, V., Jiang, R., Hoogendoorn, S., 2018. {A geometric
  Brownian motion car-following model: towards a better understanding of
  capacity drop}. {Transportmetrica B} {in press}.

\bibitem[{Zhong et~al.(2013)Zhong, Sumalee, Pan, and Lam}]{Zhong2013}
Zhong, R.~X., Sumalee, A., Pan, T.~L., Lam, W. H.~K., 2013. {Stochastic cell
  transmission model for traffic network with demand and supply uncertainties}.
  {Transportmetrica A: Transport Science} 9, 567--602.

\bibitem[{Zhou et~al.(2017)Zhou, Qu, and Li}]{Zhou2017}
Zhou, M., Qu, X., Li, X., 2017. {A recurrent neural network based microscopic
  car following model to predict traffic oscillation}. {Transportation Research
  Part C} 84, 245--264.

\end{thebibliography}


\end{document}